\documentclass{jams-l}

\usepackage{enumerate}
\usepackage{graphicx}
\usepackage{xcolor}
\usepackage{tikz}
\usepackage{comment}

\usepackage{tikz}
\usetikzlibrary{arrows}

\usepackage{hyperref}       % hyperlinks

\newtheorem{theorem}{Theorem}
\newtheorem*{theorem*}{Theorem}
\newtheorem{lemma}[theorem]{Lemma}

\newtheorem*{definition}{Definition}
\newtheorem{claim}[theorem]{Claim}
\newtheorem{obs}[theorem]{Observation}

\newtheorem*{remark*}{Remark}
\newtheorem*{claim*}{Claim}

\newtheorem*{remark}{Remark}

\newcommand{\perm}{\mathsf{perm}}
\renewcommand{\det}{\mathsf{det}}

\newcommand{\X}{\mathcal{X}}
\newcommand{\Y}{\mathcal{Y}}
\newcommand{\F}{\mathbb{F}}

\newcommand{\N}{\mathbb{N}}
\newcommand{\dc}{\mathsf{dc}}

\newcommand{\FF}{\mathbb{F}}
\newcommand{\NN}{\mathbb{N}}

\numberwithin{equation}{section}

\begin{document}

\title{The Algebraic Cost of a Boolean Sum}


%    author one information

%% Authors are commented since STACS process is double-blind
\author{Ian Orzel}
\address{Department of Computer Science, The University of Copenhagen}

\author{Srikanth Srinivasan}\thanks{Srikanth Srinivasan was funded by the European Research Council (ERC) under grant agreement no. 101125652 (ALBA). Amir Yehudayoff is supported by a DNRF chair grant, and partially supported by BARC. Sébastien Tavenas is supported by ANR-22-CE48-0007.}
\address{Department of Computer Science, The University of Copenhagen}

\author{S\'{e}bastien Tavenas}
\address{Charg\'e de Recherche CNRS au Laboratoire LAMA de \'lUniversit\'e Savoie Mont Blanc}

\author{Amir Yehudayoff}
\address{Department of Computer Science, The University of Copenhagen,
and Department of Mathematics, Technion-IIT}
%\email{amir.yehudayoff@gmail.com}

\begin{abstract}
It is a well-known fact that the permanent polynomial is complete for the complexity class VNP, and it is largely suspected that the determinant does not share this property, despite its similar expression.
We study the question of why the VNP-completeness proof of the permanent fails for the determinant.
We isolate three fundamental properties that are sufficient to prove a polynomial sequence is VNP-hard, of which two are shared by both the permanent and the determinant.
We proceed to show that the permanent satisfies the third property, which we refer to as the ``cost of a boolean sum," while the determinant does not, showcasing the fundamental difference between the polynomial families.
We further note that this differentiation also applies in the border complexity setting and that our results apply for counting complexity.
\end{abstract}

\begin{comment}
\begin{abstract}
The P versus NP problem is about the computational power
of an existential $\exists_{w \in \{0,1\}^n}$ quantifier.
The VP versus VNP problem is about the power of a boolean sum $\sum_{w \in \{0,1\}^n}$ operation.
We study the power of a single boolean sum $\sum_{w \in \{0,1\}}$,
and prove that in some cases the cost of eliminating this sum is large.
This identifies a fundamental difference between the permanent and the determinant.
This investigation also leads to a slightly simpler proof we are aware of that the permanent is VNP-complete.
\end{abstract}
\end{comment}

\maketitle

\begin{comment}
\color{red}
We often want to represent polynomials in the following form: we select disjoint sets $\Sigma_{S_1} \dots \Sigma_{S_\ell}f_j(a) = g_i$, where $j$ and $\ell$ are polynomial in $i$ and $a$ contains elements of $\X \cup \FF$, where each element of $S_1, \ldots, S_\ell$ appears in $a$ exactly once. This appears to be a kind of projection, and I believe that it should be defined so that it can be referred to easily. Perhaps such a projection already has a name? Maybe we could call it a $\Sigma$-projection for now?
\marginpar{\tiny SS: Let's not call it a projection since this refers to a substitution (and not a sum). Iterated Boolean sum?}
\color{black}
\end{comment}

%
%\bibliographystyle{amsplain}
%\bibliography{bib.bib}

\section{Introduction}

In an attempt to introduce an algebraic approach to solving the classical P versus NP problem, \cite{valiant1979completeness} introduces algebraic circuits and the complexity classes VP and VNP.
Central to this approach were two fundamental objects: the determinant and the permanent.
The determinant belongs to the class VP (it is, in fact, complete for a restriction of this class, VBP, or even VP if we consider quasi-polynomial projections), and the permanent is complete for the class VNP, as proven in \cite{valiant1979completeness}.
We suspect that the determinant is not VNP-complete, meaning that Valiant's proof of the VNP-hardness of the permanent should fail when the determinant is used in its place.
It has been unclear exactly why this is the case, although one would expect it to be related to the construction of a particular graph-theoretical gadget, known as the \emph{iff-coupling} gadget, being impossible in the model of the determinant. However, this is a somewhat narrow and technical reason.

This paper seeks to develop a more general understanding of the reasons this VNP-hardness proof fails when applied to the determinant.
In exploring this question, we simplify the proof for the VNP-completeness of the permanent and isolate the exact property that the determinant lacks which causes this VNP-hardness proof to fail (assuming all other properties are the same).

Recall that the VNP-completeness proof of the permanent begins by taking a formula and summing it over all $0$-$1$ substitutions of some of its variables (representing an element of a VNP polynomial sequence).
This formula is then transformed into a similarly-sized permanent. We then iteratively remove each $0$-$1$ substitution of a particular variable by modifying the matrix we take the permanent of, only increasing its size by a small amount.
This iterative removal is the key piece of the proof that makes it work in the case of the permanent, but, to the authors' knowledge, there does not exist a proof that this step breaks down for the determinant.
Studying this step in its full generality is difficult, as the given matrix could have a complex structure. %, and the definition of a "negligible amount" was not well-defined.
In this paper, we simplify the VNP-completeness proof by showing that we can, without loss of generality, assume that the matrix has exactly two instances of the variable that we seek to remove. Using a small (standard) gadget construction, a boolean sum over two variables can be simulated by a permanent of a matrix that is bigger \emph{only by an additive constant}, which shows that the permanent is VNP-complete. 
We can then show that the proof fails for the determinant by proving that the last step of this process must incur a super-constant (in fact, a constant \emph{multiplicative factor}) increase in the size of the matrix.

We generalize this concept of ``iterated substitutions" to arbitrary polynomials through the boolean sum.
Given a polynomial $f \in \FF[\X]$ and a subset of its variables $S \subseteq \X$, we can represent the boolean sum, which we denote by $\Sigma_Sf$, as the sum of taking $f$ with each variable in $S$ set to $0$ and the same with respect to $1$.
We can now rephrase this step of the proof as converting a boolean sum over an arbitrary-sized set of variables to multiple boolean sums over two variables.

In this paper, we have isolated the following properties as sufficient for proving the VNP-hardness of a polynomial:
\begin{enumerate}
    \item The polynomial sequence is VF-hard, i.e., a polynomial computed by a formula can be represented as a projection of a polynomial in the sequence (we can replace arguments by variables or constants), indexed by a polynomial in the size of the formula,
    \item A boolean sum over many variables can be rewritten as a sequence of boolean sums, each over two variables, and
    \item A boolean sum of a polynomial over two variables can be rewritten as a simple projection of itself, whose index is only increased by a fixed constant.
\end{enumerate}

With these properties in mind, the main results of our paper can be summarized as follows:
\begin{itemize}
    \item A polynomial sequence satisfying properties (1), (2), and (3) is VNP-hard.
    \item Both the determinant and the permanent satisfy properties (1) and (2).
    \item The permanent satisfies property (3), but the determinant does not.
\end{itemize}
In short, we have isolated the simple property that differentiates the determinant from the permanent in our ability to prove VNP-completeness.

\subsection{Related Work}
The focus of this paper is on studying the noteworthy differences in properties between the permanent and the determinant polynomials, which is an important research direction, as it seems that the separation between VP and VNP is encoded in the differences between these polynomials.
It is therefore that the differences in these polynomials have been well-studied.

One significant line of study were those seeking to find lower bounds on the size of a determinant that simulates the permanent (determinantal complexity).
In a result of von zur Gathen~\cite{gathen}, a lower bound was determined by studying the algebraic-geometric properties of the determinant and the permanent.
Specifically, there was a difference between the dimension of the spaces of singular points of the determinant and permanent (points where a polynomial and all of its first-order partial derivatives vanish).
Then, in a result of Mignon and Ressayre~\cite{mignon2004quadratic} (see also~\cite{cai2010quadratic}), the bound was strengthened by studying the rank of the Hessian of these polynomials at vanishing points.
It was shown that all zeros of the determinant have a very small Hessian rank, whereas the permanent has zero points with very large Hessian rank.

The work of Landsberg and Ressayre~\cite{LR} gives an exponential lower bound on the determinantal complexity of the permanent in a restricted setting.
They showed that any reductions from the permanent to the determinant that preserve the underlying symmetries must have exponential cost.

Recent work of Hrub\v{e}s and Joglekar~\cite{HJ} points to a distinction in reductions to \emph{read-once} determinants and permanents, which forces each variable to appear in the corresponding matrix at most once.
Here, they show that every multilinear polynomial can be represented as a read-once permanent, whereas representations using the determinant could require a variable to appear at least $\omega(\sqrt{n}/\lg{n})$ times.

\subsection{Our Results}
We have isolated three properties that are fundamental to understanding our VNP-completeness results for the permanent.
We will begin by formally stating these properties.

\begin{definition}
We say that a polynomial sequence $f = (f_i)_{i \in \NN} \in \FF[\X]$ \emph{simulates formulas} if, for every polynomial $g \in \FF[\X]$ that can be computed by a formula of size $s$, $g$ is a projection of some $f_i$, where $i$ is polynomial in $s$. In other words, $f$ is VF-hard.

\begin{comment}
\textcolor{red}{Def 1:}
We further say that $k \in \NN$ variable instances are sufficient for boolean sums of $f$ if, for every fixed $i \in \NN$ and $S \subseteq \X$, there exists disjoint sets $T_1, \ldots, T_m \subseteq \X$ and an $n \in \NN$ such that $\Sigma_Sf_i$ is a simple projection of $\Sigma_{T_1} \dots \Sigma_{T_m}f_{i+n}$, where $|S_1| = \dots = |S_m| = k$ and $m, n$ are linear in $k$.
\end{comment}

We further say that \emph{$k \in \NN$ variable instances are sufficient for boolean sums of $f$} if, for every boolean sum $\Sigma_{S_1} \dots \Sigma_{S_\ell}f_i$ [see (\ref{form:bool_sum})], there is an $\ell'$ and $m$ that are polynomial in $i$ and $\ell$ such that this sum can be rewritten as $\Sigma_{T_1} \dots \Sigma_{T_{\ell'}}f_m$, where the size of each $T_i$ is $k$.

Finally, we define the algebraic additive cost of boolean summation over $f$ via the map $\alpha_f : \N \to \N\cup\{\infty\}$,
defined as follows: for each $s \in \N$, let $\alpha_f(s)$ be the minimum integer\footnote{If there is
no such integer, then it is infinite.}
$k$ so that, for every $n$ and $S \subseteq \X_n$ of size $|S| \leq s$,
the polynomial $\Sigma_S f_n$ is a simple projection of $f_{n+k}$.
In other words, $\alpha_f(s)$ is the additive $f$-cost of performing a boolean sum
over $s$ variables. 

We can further naturally move the definition of $\alpha_f$ into the border complexity setting, which we denote by $\underline{\alpha_f}$.
Formally, we let $\underline{\alpha_f}(s)$ be the minimum integer $k$ so that there is a border sequence $F = (F_i)_{i \in \NN}$ over $\FF(\epsilon)[\X]$ such that, for every $i \in \NN$, there exist polynomials $g_i\in \FF[\X\cup \{\epsilon\}]$ and $Q_i\in \FF[\epsilon]$ and $M\in \mathbb{Z}$ such that
\[
F_i = \frac{\epsilon^{M}\cdot f_i + \epsilon^{M+1}\cdot g_i}{1+\epsilon Q_i(\epsilon)},
\]
 and furthermore, $\alpha_F(s) = k$. The algebraic expression captures our intuition that $\lim_{\epsilon \rightarrow 0} F_i = f_i.$

\end{definition}

We can now introduce our main results that use these properties. We first have that these properties are sufficient for VNP-hardness.
\begin{theorem}\label{thm:vnp_hard}
If, for a fixed $k \in \NN$, $f = (f_i)_{i \in \NN} \in \FF[\X]$ is such that
\begin{enumerate}
    \item $f$ simulates formulas,
    \item $k$ variables instances are sufficient for boolean sums over $f$, and
    \item $\alpha_f(k) < \infty$,
\end{enumerate}
then $f$ is VNP-hard.

Further, if the projections and reductions associated with these three properties can be completed in polynomial time, then $f$ is \#P-hard.
\end{theorem}

We use $\perm$ and $\det$ to denote the permanent and determinant polynomials respectively. It is a well-known fact that $\perm$ and $\det$ can simulate formulas.
\begin{lemma}[\cite{valiant1979completeness}]
The $\perm$ and $\det$ polynomial sequences can simulate formulas.
\end{lemma}

Our main deviation from the well-known VNP-hardness proof is through the decrease of the number of instances of the variables we sum over to two. This idea is stated in the following lemma.
\begin{lemma}\label{lem:two_vars}
Two variables instances are sufficient for boolean sums over the $\perm$ and $\det$ polynomials.
\end{lemma}

Finally, the main result of this paper is through differentiating the determinant from the permanent, which is stated in the following theorem.
\begin{theorem}\label{thm:alpha_two}
We have that $\alpha_\perm(2) \le 3$, but $\alpha_\det(2) = \infty$.
We further observe that $\underline{\alpha_\det}(2) = \infty$.
\end{theorem}

\subsection{Structure of Paper}
First, in Section \ref{sec:prelim}, we will introduce our notation and definitions.
Then, in Section \ref{sec:det}, we will prove Theorem \ref{thm:alpha_two}, which will differentiate the permanent from the determinant. We will do this using a well-known result of Mignon and Ressayre\cite{mignon2004quadratic} and through a careful selection of matrices.
In Section \ref{sec:two_sum}, we will prove Lemma \ref{lem:two_vars}, showing that we can transform arbitrarily-sized boolean sums to those of size two.
In Section \ref{sec:vnp_complete}, we provide the proof of Theorem \ref{thm:vnp_hard}, which is (in our opinion) a slightly simpler and more modular proof of the VNP-completeness of the permanent.

\section{Preliminaries}\label{sec:prelim}
Let $\FF$ be a field of characteristic not two. Let $\X$ be an infinite set of indeterminants and $\FF[\X]$ be the polynomial ring containing polynomials that use a finite number of indeterminants in $\X$. We will write $(f_i)_{i \in \NN} \in \FF[\X]$ to denote a sequence a polynomials (we will somtimes simply write it as $(f_i)$), and it is assumed that $f_i \in \FF[\X_i]$, where $\X_i \subseteq \X$ has size polynomial in $i$. Given a set $S$, we may use $S^{\X_i}$ to denote an assignment of the variables in $\X_i$ to elements of $S$, so that we can then see $f_i(a)$ to represent evaluating $f_i$ at the point $a \in S^{\X_i}$.

\begin{comment}
Sometimes, we may say that some number $i$ is polynomial in quantities $n_1, \ldots, n_k$. By this, we mean there is a polynomial $f \in \NN[x_1, \ldots, x_k]$ such that, when given $n_1, \ldots, n_k$, $i = f(n_1, \ldots, n_k)$.
\end{comment}

We will use the typical definitions for circuits, formulas, the complexity classes VF, VP, and VNP, and hardness within these classes. For more information, see~\cite{valiant1979completeness}.
\begin{comment}
The class VP captures efficient computations of polynomials. 
The class VNP is defined by adding a boolean sum
to the class VP. 
\end{comment}

A fundamental tool in computational complexity theory
is reduction between problems~\cite{karp1975computational}.
Reductions allow us to compare the complexities of different problems. 
In the algebraic setting, reductions are projections.
The polynomial $f \in \F[\X]$ is a projection of $g \in \F[\Y]$
if there is a map $\phi : \Y \to \X \cup \F$
so that $f = \lambda (g \circ \phi)$, where \(\lambda \in \mathbb{F}\setminus \{0\}\).
The projection is called {\em simple} if for each $x \in \X$
the cardinality of $\phi^{-1}(x)$ is at most one.

For boolean summation, we use the following notation. 
For $S \subseteq \X$,
denote by $f_S$ the polynomial in the variables $\X \cup \{y\}$, where $y$ is a new variable,
that is obtained from $f$ by the substitution $x= y$ for all $x \in S$.
For example, $(x_1+x_2^2+x_3)_{\{x_1,x_2\}} = y+y^2 + x_3$.
Denote by $\Sigma_S f$ the polynomial 
\begin{equation}\label{form:bool_sum}
\Sigma_S f = f_S |_{y=0} + f_S|_{y=1}.
\end{equation}
For example, $\Sigma_{\{x_1,x_2\}} (x_1+x_2^2+x_3) = 2 + 2 x_3$.
The variables in $S$ do not appear in $f_S$ and $\Sigma_S f$.

For two polynomial sequences $g = (g_i)_{i \in \NN}$ and $f = (f_i)_{i \in \NN}$, we say that $g$ is an iterated boolean sum of $f$ if, for every $i \in \NN$, there are disjoint sets $S_1, \ldots, S_\ell \subseteq \X$ such that $g_i = \Sigma_{S_1} \dots \Sigma_{S_\ell}f_j$, where $\ell$ and $j$ are polynomial in $i$.
Observe that a polynomial sequence being in VNP is equivalent to there being a sequence in VF each element in the VNP sequence is an iterated boolean sum of a (polynomially-indexed) element in the VF sequence.

The two central polynomial families in algebraic complexity are
the permanent and the determinant, defined by,
$$\perm(X) = \perm_n(X) = \sum_{\pi} \prod_i x_{i,\pi(i)}$$
and
$$\det(X) = \det_n(X) = \sum_\pi (-1)^{\mathsf{sign}(\pi)} \prod_i x_{i,\pi(i)},$$
where $X = (x_{i,j})$ is an $n \times n$ matrix of variables,
the sum is over permutations $\pi$ of $[n]$,
and the product is over $i \in [n]$.
The permanent and determinant can be equivalently defined as the sum of the (signed) weights of cycle covers over a graph, whose adjacency matrix is given;
see~\cite{mahajan1999determinant}.
The permanent is VNP-complete, so that VP $\neq$ VNP if and only if $\perm$ is not in VP~\cite{valiant1979completeness}.
It also appears in many places in computer science
because it encapsulates many counting problems
(see e.g.~\cite{valiant1979complexity,toda1991pp,aaronson2011linear} and references within);
by definition, it counts the number of perfect matchings in a bipartite graph.
The VP versus VNP problem is
(essentially) about the relationship between the permanent and
the determinant (see~\cite{valiant1979completeness,mulmuley2011p,malod2008characterizing} and
references within).

The determinant defines a complexity measure;
the determinantal complexity $\dc(f)$ of a polynomial $f$ is
the minimum $k$ so that $f$ is a projection of $\det_k$.
It is known that $\dc(f)<\infty$ for all $f$.
The VP versus
VNP problem is roughly captured by {\em what is $\dc(\perm_n)$}?
Proving that $\dc(\perm_n)$ is super-polynomial in $n$ is one of the major
open problems in computer science. 
The best lower bound we know is $\dc(\perm_n) \geq \frac{n^2}{2}$;
see~\cite{mignon2004quadratic,cai2010quadratic}.

Given a complexity measure $L : \FF[\X] \rightarrow \NN$, we can define its corresponding ``border" measure, denoted $\underline{L}$, to be the polynomial with the smallest $L$ measure that approximates the original polynomial.
More formally, given an $f \in \FF[\X]$, we define $\underline{L}(f)$ to be the smallest $k$ such that there exist $F \in \FF(\epsilon)[\X]$, $g\in \FF[\X\cup \{\epsilon\}]$, $Q\in \FF[\epsilon]$, and $M\in \mathbb{Z}$ such that
\[
F = \frac{\epsilon^{M}f + \epsilon^{M+1}g}{1+\epsilon Q(\epsilon)},
\]
and furthermore, $L(F) = k$.

In the field of counting complexity, we study the complexity of functions that ``count" some quantity, mapping $\varphi : \{0, 1\}^* \rightarrow \NN$.
Recall that we define the class NP as the set of decision problems, denoted by $\varphi : \{0, 1\}^* \rightarrow \{0, 1\}$, such that there is a Turing Machine taking in the input and a special "witness" string such that the Turing Machine runs in polynomial time and a string is in the language if and only if there is at least one "witness" string for which the Turing Machine accepts it.
We then define the complexity class $\sharp$P similarly, but instead we require the related Turing Machine to have the number of accepting witness strings to a given input to be the same as that defined by $\varphi$.
The quintessential $\sharp$P-complete problem is $\sharp3$-SAT, which, given a $3$-CNF formula, counts the number of satisfying assignments.

\section{Cost of Boolean Summation}\label{sec:det}
This section will focus on proving Theorem \ref{thm:alpha_two}. We will separate the proof of this theorem into two sub-claims. The first is used to show that $\alpha_\text{det}(2) = \infty$.
\begin{claim}\label{clm:det}
For any $n \in \NN$, we have that
\[
\dc(\Sigma_{\{x_{1,1}, x_{2,2}\}}\det_n) \ge 2(n-2).
\]
\end{claim}

Next, we will prove that $\alpha_\perm(2) \le 3$ through another claim.
\begin{claim}\label{clm:per}
There exists a matrix $E \in (\X \cup \FF)^{(n+3) \times (n+3)}$ such that
\[
\Sigma_{\{x_{1,1}, x_{1,2}\}}\perm = \perm(E).
\]
\end{claim}
Observe that, do to row and column swaps (and the fact that, if the two indeterminants were on the same row or column, the claim would be trivially true), these claims imply what we want.

\subsection{Lower Bound on Determinantal Boolean Sums}
We now prove that $\alpha_\det(2) = \infty$ by proving Claim \ref{clm:det}.
The proof relies on a mechanism 
developed by Mignon and Ressayre to lower bound the determinantal complexity~\cite{mignon2004quadratic}.
For a polynomial $f \in \F[\X]$,
denote by $H_f$ the $|\X| \times |\X|$ Hessian matrix of $f$
defined by
$$(H_f)_{x,x'} = \frac{\partial^2}{\partial x \partial x'} f.$$

\begin{lemma}[\cite{mignon2004quadratic}; see also Lemma 13.3 in~\cite{chen2011partial}]
Let $f$ be a polynomial in the variables $x = (x_1,\ldots,x_n)$
and let $a \in \F^n$ be so that $f(a)=0$ 
then $$\dc(f) \geq \frac{\mathsf{rank}(H_f|_{x=a})}{2}.$$
\end{lemma}

The theorem follows by considering the polynomial
$$f = \Sigma_{\{x_{1,1},x_{2,2}\}} \det_n$$ and
the $n \times n$ matrix
$$A = \left[ 
\begin{array}{ccc}
 & \frac{1}{2} & \\ 
1 &  &  \\
 &  & I \\
\end{array}
\right] , $$
where $I$ is the $(n-2) \times (n-2)$ identity matrix. 
In the two claims below (which complete the proof), let $H = H_f$.

\begin{claim}
\label{clm:0}
$f(A)=0$.
\end{claim}

\begin{proof}[Proof of Claim~\ref{clm:0}]
\begin{align*}
f(A)
& = \det \left( 
\left[ 
\begin{array}{ccc}
 & \frac{1}{2} & \\ 
1 &  &  \\
 &  & I \\
\end{array}
\right]  \right) 
+ 
\det \left( 
\left[ 
\begin{array}{ccc}
1 & \frac{1}{2} & \\ 
1 & 1 &  \\
 &  & I \\
\end{array}
\right]  \right) \\
& = - \frac{1}{2}+ 1 - \frac{1}{2} = 0. \qedhere
\end{align*}
\end{proof}

\begin{claim}
\label{clm:1}
$\mathsf{rank}(H|_{X=A}) \geq 4(n-2)$. 
\end{claim}

\begin{proof}[Proof of Claim~\ref{clm:1}]
Focus on the $4(n-2)$ variables \(\mathcal{V}_i =\{x_{1,i}, x_{2,i}, x_{i,1}, x_{i,2}\}\) for $i>2$.
Furthermore, we consider \(G\) to be the $4(n-2) \times 4(n-2)$ sub-matrix of \(H\) whose rows and columns are labeled by variables in \(\bigcup_{i>2} \mathcal{V}_i\).
Let $f_0 = \det_n|_{x_{1,1}=x_{2,2} = 0}$
and $f_1 = \det_n|_{x_{1,1}=x_{2,2} = 1}$,
so that $f = f_0+f_1$.

An entry in the Hessian of $\det$ is plus/minus the determinant of the \((n-2)\)-minor
obtained by deleting the corresponding rows and columns from $X$
(or zero if the two variables share a row or a column). 

For convenience, denote by $H_{k\ell mp}$ the $(x_{k,\ell},x_{m,p})$-entry in $H$ (and in \(G\)).

We first claim that after substituting $A$,
the sub-matrix \(G\) of $H$ has $n-2$ blocks, each of size $4 \times 4$:
$$\left[ 
\begin{array}{cccc}
G_3 &  & & \\ 
 & G_4 &   & \\
 &  & \ddots & \\ 
&  & & G_n \\ 
\end{array}
\right]$$
where for $i > 2$
$$G_i = \left[ 
\begin{array}{cccc}
H_{1i1i} & H_{1i2i} & H_{1ii1} & H_{1ii2}\\ 
H_{2i1i} & H_{2i2i} & H_{2ii1} & H_{2ii2}\\ 
H_{i11i} & H_{i12i} & H_{i1i1} & H_{i1i2}\\ 
H_{i21i} & H_{i22i} & H_{i2i1} & H_{i2i2}\\ 
\end{array}
\right].$$
Indeed, for $i \neq j$, the matrix obtained from $A$ by deleting row $i$
and column $j$ has a zero row
so its determinant is zero. 
It follows for \(k,\ell \in \{1,2\}\) that
$$(H_f)_{ik\ell j}|_{X=A}
= (H_{f_0})_{ik\ell j}|_{X=A} + (H_{f_1})_{ik\ell j}|_{X=A}
= 0 + 0 = 0.$$
A similar argument holds for the other entries $H_{ikj\ell}$, $H_{ki\ell j}$, and $H_{kij\ell}$. 

Now fix $i$ and focus on $G_i$. 
Entries of the form $1i1i$ and $1i2i$ are also zero.
So, we are left with 
$$G_i = \left[ 
\begin{array}{cccc}
 &  & H_{1ii1} & H_{1ii2}\\ 
 &  & H_{2ii1} & H_{2ii2}\\ 
H_{i11i} & H_{i12i} & & \\ 
H_{i21i} & H_{i22i} &  & \\ 
\end{array}
\right].$$
The matrix $G_i$ is symmetric so it is enough to understand 
$$\left[ 
\begin{array}{cc}
 H_{1ii1} & H_{1ii2}\\ 
H_{2ii1} & H_{2ii2}\\ 
\end{array}
\right].$$
This matrix is a sum of two matrices (one from $f_0$ and one from $f_1$).
For $f_0$, this matrix is 
$$\left[ 
\begin{array}{cc}
 & 1 \\ 
\frac{1}{2} &  \\ 
\end{array}
\right].$$
For $f_1$, this matrix is 
$$\left[ 
\begin{array}{cc}
-1 & 1 \\ 
\frac{1}{2} &  -1 \\ 
\end{array}
\right].$$
The sum of the two matrices is  
$$\left[ 
\begin{array}{cc}
-1 & 2 \\ 
1 & -1 \\ 
\end{array}
\right],$$
which always has rank two.
Rolling back
\begin{equation*}
\mathsf{rank} \left( \left[ 
\begin{array}{cccc}
G_3 &  & & \\ 
 & G_4 &   & \\ 
 &  & \ddots & \\ 
&  & & G_n \\ 
\end{array}
\right]\right) = 4(n-2). \qedhere
\end{equation*}
\end{proof}

\begin{obs}
We note that this result also applies in the border complexity setting.
Specifically, it also follows that $\underline{\dc}(\Sigma_{\{x_{1,1}, x_{2,2}\}}\det_n) \ge 2(n-2)$.
This fact follows from \cite{grochow2015bordermr}, where it is shown that the lower bound from \cite{mignon2004quadratic} of the determinantal complexity of the permanent also applies in the border complexity setting.
Although the proof presented there is specifically for the permanent polynomial, it is easy to see that it also works if you use any arbitrary homogeneous polynomial.
We should observe that $\Sigma_{\{x_{1,1}, x_{2,2}\}}\det_n$ is not homogeneous, but one can easily see that we can apply the result to the natural homogenization of the polynomial and get the desired result.
\end{obs}

\subsection{Upper Bound on Permanental Boolean Sum}
In this section, we will prove that $\alpha_\perm(2) \le 3$ by proving Claim \ref{clm:per}. Specifically, we will select the matrix $E$ to be
\begin{align*}
E =
\left[ 
\begin{array}{ccc|cccc}
1 & 1 & 1 &&&& \\ 
1 & 0 & -1 & 1 &&&  \\
-\tfrac{1}{2} & \tfrac{1}{2} & \tfrac{3}{2} &&1&&  \\
\hline
 & 1 &  & 0 &x_{1,2} &x_{1,3}& \cdots  \\
 &  &  1 & x_{2,1} & 0 &x_{2,3}&  \\
 &  &   & x_{3,1} & x_{3,2} &x_{3,3}&  \\
 & &  & \vdots & && \ddots 
\end{array}
\right].
\end{align*}

\begin{proof}[Proof of Claim \ref{clm:per}]
To see this, denote by $X[b]$ the matrix
$X$ after the substitution $x_{1,1}=x_{2,2} =b$ for $b \in \{0,1\}$.
For $S \subseteq \{1,2\}$, denote by $X[b]_{-S}$
the matrix $X[b]$ after deleting the rows and columns in $S$.
We have
\begin{align*}
\perm(X[1])
& = \perm \left( \left[ 
\begin{array}{cccc}
1 & x_{1,2} &x_{1,3}&  \\
 x_{2,1} & 1 &x_{2,3}&  \\
 x_{3,1} & x_{3,2} &x_{3,3}&  \\
& && \ddots \\
\end{array}
\right] \right) \\ 
& = 
\perm(X[0])
+ \perm(X[0]_{-\{1\}})
+ \perm(X[0]_{-\{2\}})
+ \perm(X[0]_{-\{1,2\}}) .
\end{align*}

Now, consider a permutation $\pi$ such that the associated value of the permutation on $E$ is nonzero. We can then notice that such a permutation falls into one of four categories:
\begin{itemize}
    \item $\pi({\{1, 2, 3\}}) = \{1,2,3\}$,
    \item $\pi(2) = 4$ and $\pi({\{1, 3\}}) = \{1,3\}$,
    \item $\pi(3) = 5$ and $\pi({\{1, 2\}}) = \{1,2\}$, or
    \item $\pi(2) = 4$, $\pi(3) = 5$, and $\pi(1) = 1$.
\end{itemize}
Let $U$ be the $3 \times 3$ upper-left block
\begin{align*}
    U = \begin{bmatrix}
        1 & 1 & 1 \\
        1 & 0 & -1 \\
        -\frac{1}{2} & \frac{1}{2} & \frac{3}{2}
    \end{bmatrix}.
\end{align*}
For $S,T \subseteq \{1,2,3\}$ with $\lvert S\rvert = \lvert T\rvert$, we denote by $U_S^T$ the sub-matrix of $U$ obtained by keeping the rows of $S$ and the columns of $T$. The matrix $U$ is chosen such that
\begin{align*}
    & \perm(U_{\{1,2\}}^{\{1,3\}}) = \perm(U_{\{1,3\}}^{\{1,2\}}) = 0, \\
    & \perm(U_{\{1\}}^{\{1\}}) = \perm(U_{\{1,2\}}^{\{1,2\}}) = \perm(U_{\{1,3\}}^{\{1,3\}}) = 1, \\ 
    & \perm(U) = 2.
\end{align*}

In the first of the above 4 cases, it must be that $\pi({\{4, 5, \ldots, m\}}) = \{4,5,\ldots, m\}$. Therefore, the sum of the weights of these permutations corresponds to
\[
\perm \left( \left[
    \begin{array}{c|c}
      U & \\
      \hline
       & X[0]
    \end{array}
    \right] \right) = 2\perm(X[0]).
\]

In the second case, it must be that $\pi(4) = 2$ and $\pi({\{5, 6, \ldots, m\}}) = \{5,6,\ldots, m\}$. Thus, this case corresponds to
\[
\perm \left( \left[
    \begin{array}{c|c}
      U_{\{1,3\}}^{\{1,3\}} & \\
      \hline
       & X[0]_{-\{1\}}
    \end{array}
    \right] \right) = \perm(X[0]_{-\{1\}}).
\]

In the third case, we have that $\pi(5) = 3$ and $\pi({\{4, 6,7 ,\ldots, m\}}) = \{4,6,7,\ldots, m\}$. Thus, this case corresponds to
\[
\perm \left( \left[
    \begin{array}{c|c}
      U_{\{1,2\}}^{\{1,2\}} & \\
      \hline
       & X[0]_{-\{2\}}
    \end{array}
    \right] \right) = \perm(X[0]_{-\{2\}}).
\]

Finally, in the fourth case, it must be that $\pi(4), \pi(5) \in \{2, 3\}$ and $\pi({\{6,7 ,\ldots, m\}}) = \{6,7,\ldots, m\}$. This case then corresponds to
\[
\perm \left( \left[
    \begin{array}{c|c}
      U_{\{1\}}^{\{1\}} & \\
      \hline
       & X[0]_{-\{1,2\}}
    \end{array}
    \right] \right) = \perm(X[0]_{-\{1,2\}}).
\]

Summing it together,
\begin{equation*}
    \perm(E) = \perm(X[0]) + \perm(X[1]). \qedhere
\end{equation*}
% Such terms can also lead to zero, like 
% \begin{align*}
% \perm \left( \left[ 
% \begin{array}{cc|c}
% 1 & 1  & \\ 
% -\tfrac{1}{2} & \tfrac{1}{2}  &  \\
% \hline
%     && \ddots \\
% \end{array}
% \right] \right) = 0
% \end{align*}
% (where we deleted row $2$ and column $3$,
% and the other $1$'s are replaced by $0$'s).
% We are left with four terms. 
% Keeping all $1$'s:
% \begin{align*}
% \perm(X[0]_{-\{1,2\}}).
% \end{align*}
% Not keeping all $1$'s:
% \begin{align*}
% \perm \left( \left[ 
% \begin{array}{ccc|cccc}
% 1 & 1 & 1 &&&& \\ 
% 1 & & -1 &  &&&  \\
% -\tfrac{1}{2} & \tfrac{1}{2} & \tfrac{3}{2} &&&&  \\
% \hline
%  &  &  & &x_{1,2} &x_{1,3}&  \\
%  &  &   & x_{2,1} & &x_{2,3}&  \\
%  &  &   & x_{3,1} & x_{3,2} &x_{3,3}&  \\
%  & &  & & && \ddots \\
% \end{array} 
% \right] \right) = 2 \perm(X[0]) .
% \end{align*}
% Keeping the $1$'s in column two and row two:
% \begin{align*}
% \perm \left( \left[ 
% \begin{array}{cc|ccccc}
% 1  & 1 &&& \\ 
% -\tfrac{1}{2}  & \tfrac{3}{2} &&&  \\
% \hline
%    &    & &x_{2,3}&  \\
%    &    & x_{3,2} &x_{3,3}&  \\
%   &   & && \ddots \\
% \end{array}
% \right] \right) = \perm(X[0]_{-\{1\}}) .
% \end{align*}
% Keeping the $1$'s in column three and row three:
% \begin{align*}
% \perm \left( 
% \left[ 
% \begin{array}{cc|ccc}
% 1 & 1  &&& \\ 
% 1 &  &  &&  \\
% \hline
%  &    &  &x_{1,3}&  \\
%  &     & x_{3,1} & x_{3,3}&  \\
%  &   & & & \ddots \\
% \end{array}
% \right] \right) = \perm(X[0]_{-\{2\}}) .
% \end{align*}
% Summing it together,
% \begin{equation*}
% \perm(E) = \perm(X[0])
% + \perm(X[1]). 
%\end{equation*}
\end{proof}

\subsection{Boolean Sum over a Single Variable}
Interestingly, considering boolean sums of two variables is necessary to separate $\perm$ from $\det$. This is becaues the boolean cost of removing one variable is the same for both polynomial sequences.

\begin{claim}
We have that $\alpha_\det(1) = \alpha_\perm(1) = 0$.
\end{claim}
\begin{proof}
\begin{align*}
& \Sigma_{\{x_{1,1}\}} \det(X) \\
& = \det \left[ 
\begin{array}{ccccc}
 0 &  x_{1,2} &  x_{1,3} & \ldots &  x_{1,n} \\
x_{2,1} & x_{2,2} & x_{2,3} & \ldots & x_{2,n} \\
x_{3,1} & x_{3,2} & x_{3,3} & \ldots & x_{3,n} \\
\vdots & \vdots & \vdots & \ddots & \vdots \\
x_{n,1} & x_{n,2} & x_{n,3} & \ldots & x_{n,n} 
\end{array}
\right]
+ \det \left[ 
\begin{array}{ccccc}
1 &  x_{1,2} &  x_{1,3} & \ldots &  x_{1,n} \\
x_{2,1} & x_{2,2} & x_{2,3} & \ldots & x_{2,n} \\
x_{3,1} & x_{3,2} & x_{3,3} & \ldots & x_{3,n} \\
\vdots & \vdots &\vdots & \ddots & \vdots \\
x_{n,1} & x_{n,2} & x_{n,3} & \ldots & x_{n,n} 
\end{array}
\right] \\
% & = 2 \det \left[ 
% \begin{array}{ccccc}
%  &  x_{1,2} &  x_{1,3} & \ldots &  x_{1,n} \\
% x_{2,1} & x_{2,2} & x_{2,3} & \ldots & x_{2,n} \\
% x_{3,1} & x_{3,2} & x_{3,3} & \ldots & x_{3,n} \\
% &&& \ldots & \\
% x_{n,1} & x_{n,2} & x_{n,3} & \ldots & x_{n,n} 
% \end{array}
% \right]
% + \det \left[ 
% \begin{array}{ccccc}
% 1 &   &   &  &   \\
%  & x_{2,2} & x_{2,3} & \ldots & x_{2,n} \\
%  & x_{3,2} & x_{3,3} & \ldots & x_{3,n} \\
% &&& \ldots & \\
%  & x_{n,2} & x_{n,3} & \ldots & x_{n,n} 
% \end{array}
% \right] \\
& = 2 \det \left[ 
\begin{array}{ccccc}
1/2 & x_{1,2} & x_{1,3} & \ldots & x_{1,n} \\
x_{2,1} & x_{2,2} & x_{2,3} & \ldots & x_{2,n} \\
x_{3,1} & x_{3,2} & x_{3,3} & \ldots & x_{3,n} \\
\vdots & \vdots &\vdots & \ddots & \vdots \\
x_{n,1} & x_{n,2} & x_{n,3} & \ldots & x_{n,n} 
\end{array}
\right] .
\end{align*}
A similar calculation holds for the permanent. 
\end{proof}

\section{Two variables suffice for VNP-hardness}\label{sec:two_sum}
The aim of this section is to prove Lemma~\ref{lem:two_vars}
that shows that summation over two variables suffices for VNP-hardness. 
We prove it for the determinant
(the proof for the permanent is similar).  
The following construction
shows how to replace a boolean summation over $m$
variables by $m$ boolean summations over two variables. 

\begin{obs}
Let $f \in \F[\X]$ and $S = \{x_1,\ldots,x_m\} \subseteq \X$.
Let $y_1,\ldots,y_m$ be $m$ new variables, and define 
$$g(y) = \Big( \prod_{i \in [m]} y_i \Big) + \Big( \prod_{i \in [m]}(1-y_i) \Big) .$$
Then
$$\Sigma_{\{x_1,y_1\}} \Sigma_{\{x_2,y_2\}} \ldots \Sigma_{\{x_m,y_m\}} g f = \Sigma_S f.$$
Indeed, on the l.h.s.\ there is a sum over $2^m$ terms,
but $g$ picks exactly two of them that perfectly
agree with the r.h.s.\
\end{obs}

The observation allows to replace a boolean sum over many variables
by several sums, each over just two.

\begin{obs}
Let $g$ be the polynomial defined above. 
Then, there is a $2m \times 2m$ matrix $A$ with entries in
$\{y_1,\ldots,y_m,0,1,-1\}$ so that %$a \leq 3m$,
each $y_i$ appears once in $A$ and
$$g = \det(A).$$
The matrix $A$ is the $m$-variate version of the following matrix:
$$ A = \left[ 
\begin{array}{cc|cc|cc|cc}
-1 & 1 & -1 & &&&&   \\
-1 & y_1 && &&&&  \\
\hline
&& -1 & 1  & -1 &&&   \\
&& -1 & y_2  &&&&  \\
\hline
&&&& -1 & 1  & -1 &  \\
&& && -1 & y_3  &&  \\
\hline
1 &&&& && -1 & 1 \\
&&&& && -1 & y_4 \\
\end{array}
\right] .$$

To see that \(\det(A)\) is the polynomial \(g\), it is easier to see the determinant as a signed sum of the weights of the cycles covers of the graph \(G_A\) whose \(A\) is the adjacency matrix (see Figure~\ref{fig:GA})
\begin{align*}
    \det(A) = \sum_{C \text{ cycles cover}} \left[ \prod_{\sigma \text{ cycle of } C} (-1)^{1+\text{length of }C}\rm{weight}(C) \right].
\end{align*}
The reader could find more details about this characterization by cycles covers in~\cite{mahajan1999determinant}. The only cycle which contains a dashed edge is exactly the union of the dashed edges. So the only cover that contains this cycle also contains the \(m\) \(y_i\)-loops. It corresponds to one cover of signed weight \(\left((-1)^{m+1}(-1)^{m-1}\right)\prod_{i=1}^m y_i = \prod_{i=1}^m y_i\). Otherwise, the covers contain no dashed edges and are given by \(m\) disjoint graphs. The signed weight is the product of them: \(\prod_i((-1)(-1)+(1)(-y_i)) = \prod_i(1-y_i)\). The sum of the covers gives exactly the polynomial \(g\).
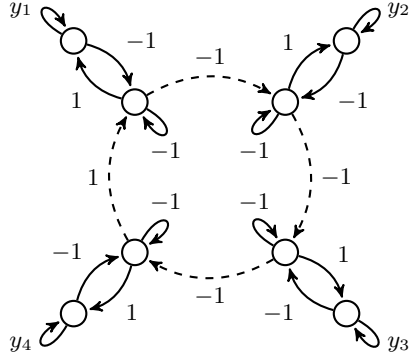
\begin{figure}
    \centering
\begin{tikzpicture}[->,>=stealth',shorten >=1pt,auto,node distance=2cm,
                    thick,main node/.style={circle,draw}]

  \node[main node] (1) {};
  \node[main node] (2) [right of=1] {};
  \node[main node] (3) [below of=2] {};
  \node[main node] (4) [below of=1] {};
  \node[main node] (1b) [above left of=1, xshift=0.6cm, yshift=-0.6cm] {};
  \node[main node] (2b) [above right of=2, xshift=-0.6cm, yshift=-0.6cm] {};
  \node[main node] (3b) [below right of=3, xshift=-0.6cm, yshift=0.6cm] {};
  \node[main node] (4b) [below left of=4, xshift=0.6cm, yshift = 0.6cm] {};

  \path[every node/.style={font=\sffamily\small}]
    (1) edge [bend left, dashed] node {$-1$} (2)
        edge [bend left] node {$1$} (1b)
        edge [out=330,in=300,looseness=15] node[below] {$-1$} (1)
    (1b) edge [bend left] node {$-1$} (1)
         edge [out=150,in=120,looseness=15] node[left] {$y_1$} (1b)
    (2) edge [bend left, dashed] node {$-1$} (3)
        edge [bend left] node {$1$} (2b)
        edge [out=240,in=210,looseness=15] node[below] {$-1$} (2)
    (2b) edge [bend left] node {$-1$} (2)
         edge [out=60,in=30,looseness=15] node[right] {$y_2$} (2b)  
    (3) edge [bend left, dashed] node {$-1$} (4)
        edge [bend left] node {$1$} (3b)
        edge [out=150,in=120,looseness=15] node[above] {$-1$} (3)
    (3b) edge [bend left] node {$-1$} (3)
         edge [out=330,in=300,looseness=15] node[right] {$y_3$} (3b)   
    (4) edge [bend left, dashed] node {$1$} (1)
        edge [bend left] node {$1$} (4b)
        edge [out=60,in=30,looseness=15] node[above] {$-1$} (4)
    (4b) edge [bend left] node {$-1$} (4)
         edge [out=240,in=210,looseness=15] node[left] {$y_4$} (4b);
\end{tikzpicture}
    \caption{The graph \(G_A\)}
    \label{fig:GA}
\end{figure}
\end{obs}

\begin{remark}
For the permanent, every $-1$ above the diagonal of the matrix $A$ should be changed into a $1$, and the $2 \times 2$ blocks into
$\left[ 
\begin{array}{cc}
-1 & 1 \\ 
1 & y_1  \\
\end{array}
\right]$.

\end{remark}

\begin{proof}[Proof of Lemma \ref{lem:two_vars}]
Let $B \in (\X \cup \FF)^{n \times n}$ and $S_1, \ldots, S_\ell \subseteq \X$ be such that $S_1, \ldots, S_\ell$ are disjoint. Then, by previous
\[
\Sigma_{S_\ell} \det(B) = \Sigma_{\{x_1, y_1\}} \dots \Sigma_{\{x_m, y_m\}}g\det(B) = \Sigma_{\{x_1, y_1\}} \dots \Sigma_{\{x_m, y_m\}}\det\left(\begin{array}{cc}
    B & 0 \\
    0 & A
\end{array}\right)
\]
Notice here that $m \le |S_\ell|$ and the size of the block matrix is at most $n + 2m \le n + 2|S_\ell|$. We can then use induction on $S_1, \ldots, S_{\ell-1}$ to conclude that we can write
\[
\Sigma_{S_1} \dots \Sigma_{S_\ell}\det(A) = \Sigma_{T_1} \dots \Sigma_{T_{\ell'}} \det(C),
\]
where $\ell' \le \sum_{i=1}^\ell|S_\ell| \le n^2$ and the size of $C$ is at most $n + 2\sum_{i=1}^\ell|S_\ell| \le n + 2n^2$.
\end{proof}

\section{The permanent is VNP-complete}\label{sec:vnp_complete}
\label{sec:perm}

This section will finally provide the proof of Theorem \ref{thm:vnp_hard}, showing that the three properties we have isolated are sufficient for VNP-hardness. The proof proceeds similarly to Valiant's VNP-completeness proof of the permanent with a few variations. This section thus also provides a (arguably) more modular proof of the VNP-completeness of the permanent.

\begin{proof}[Proof of Theorem \ref{thm:vnp_hard}]
Because $f$ simulates formulas, we can write every $n$-variate $g \in \FF[\Y]$ in VNP (see~\cite{burgisser2013completeness,shpilka2010arithmetic}) as
$$g = (\Sigma_{S_1} \ldots \Sigma_{S_\ell} f_d)(a)$$
where $a \in (\FF \cup \Y)^{\X_d}$ is a variable assignment such that $d$ is polynomial in $n$, each entry in $a$ is a variable or a field element, and
$S_1,\ldots,S_\ell$ are pairwise disjoint (implying $\ell$ is polynomial in $d$).

Because $k$ variables are sufficient for boolean sums over $f$, we can assume without loss of generality that $|S_1| = \dots = |S_\ell| = k$.
Further, as $\alpha_f(k) < \infty$, we know that we can write
\[
\Sigma_{S_\ell}f_d = f_{d+\alpha_f(k)}(b),
\]
for some simple variable assignment $b \in (\FF \cup \Y)^{\X_{i + \alpha_f(k)}}$.
Observe that, because this is a simple projection, we can easily (after modifying our sets $S_1, \ldots, S_{\ell-1}$) bring out $b$ and write $(\Sigma_{S_1} \dots \Sigma_{S_{\ell-1}}f_{d+\alpha_f(k)})(b') = g$, for some assignment $b'$.
Thus, after repeating this process $\ell$ times, we can write
\[
(\Sigma_{S_1} \dots \Sigma_{S_\ell}f_d)(a) = f_{d + \ell\alpha_f(k)}(c),
\]
where $c \in (\FF \cup \Y)^{\X_{d + \ell\alpha_f(k)}}$ is a variable assignment. Because $\ell$ and $d$ are polynomial in $n$, we can conclude that $f$ is VNP-hard.
\end{proof}

\begin{obs}
We will then conclude with the observation that this proof also works for proving $\sharp$P-hardness if each of the steps that we took in this proof can be done in polynomial time.
Consider $\sharp$3-SAT, a complete problem for $\sharp$P, where we are given a 3-CNF and must return the number of satisfying assignments it has.
Consider some 3-CNF $\varphi$ in $n$ variables, and observe that the number of possible clauses is polynomial in $n$, so we can easily encode it using one bit for each possible clause.
We will write this number as $m_n$, and one can easily see that there is a polynomial $f_n \in \FF[y_1, \ldots, y_{m_n}, x_1, \ldots, x_n]$ such that, if we provide the $y_i$ variables with the 0/1 encoding of $\varphi$ and the $x_i$ variables with a variable assignment, it will output 0/1 depending on whether this assignment is a satisfying assignment of $\varphi$.
One can easily see that $f_n$ can be computed by a formula that is polynomial in $n$.
Finally, we observe that the polynomial $g_n = \Sigma_{\{x_1\}} \dots \Sigma_{\{x_n\}}f_n$ provides the exact answer for $\sharp$3SAT when given the binary encoding of a 3-CNF.
Hence, we can carry out our previous proof from the sequence $g = (g_n)_{n \in \NN}$ to reduce it to a polynomial with the three properties (ensuring that the steps along the way can be done in polynomial time).
At the end, we conclude that we can reduce $\sharp$3-SAT to our polynomial sequence.
\end{obs}

\begin{comment}
\section{Further Research}
Separating the classes VP and VNP hinges on studying their complete problems. Specifically, one would hope to reach a result such as "the determinant is not VNP-complete".
One may hope to be able to make a statement such as "the determinant is no VNP-complete because it does not satisfy property (3)".
An interesting question that follows from this framework is: does there exist a polynomial that is VNP-complete (or simply hard) that satisfies properties (1) and (2) but does not satisfy property (3)?
In essence, do there exist other reasons that the determinant is not VNP-complete?
Many of the proofs that other polynomial sequences are VNP-complete rely on a reduction from the permanent, and it is likely that these reductions would preserve these properties.
\marginpar{\tiny SS: For the first paragraph, there may be some trivial counterexamples. For the second, it probably immediately follows from our construction. Maybe this section can be commented out for now?}

Another interesting question is that of \#P-completeness, which the permanent is a known problem for.
One can now ask: can we use this strategy to prove \#P-hardness?
Are the three properties we have outlines sufficient for proving such a result?
This would provide an interesting link to these results in algebraic complexity to those in general complexity theory.
\end{comment}

\bibliographystyle{amsplain}
\bibliography{bools}   

\providecommand{\bysame}{\leavevmode\hbox to3em{\hrulefill}\thinspace}
\providecommand{\MR}{\relax\ifhmode\unskip\space\fi MR }
\providecommand{\MRhref}[2]{%
  \href{http://www.ams.org/mathscinet-getitem?mr=#1}{#2}
}
\providecommand{\href}[2]{#2}
\begin{thebibliography}{10}

\bibitem{aaronson2011linear}
Scott Aaronson, \emph{A linear-optical proof that the permanent is \#{P}-hard}, Proceedings of the Royal Society A: Mathematical, Physical and Engineering Sciences \textbf{467} (2011), no.~2136, 3393--3405.

\bibitem{burgisser2013completeness}
Peter B{\"u}rgisser, \emph{Completeness and reduction in algebraic complexity theory}, vol.~7, Springer Science \& Business Media, 2000.

\bibitem{cai2010quadratic}
Jin-Yi Cai, Xi~Chen, and Dong Li, \emph{Quadratic lower bound for permanent vs. determinant in any characteristic}, computational complexity \textbf{19} (2010), no.~1, 37--56.

\bibitem{chen2011partial}
Xi~Chen, Neeraj Kayal, and Avi Wigderson, \emph{Partial derivatives in arithmetic complexity and beyond}, Foundations and Trends in Theoretical Computer Science \textbf{6} (2011), no.~1--2, 1--138.

\bibitem{grochow2015bordermr}
Joshua~A. Grochow, \emph{Unifying known lower bounds via geometric complexity theory}, computational complexity \textbf{24} (2015), 393–--475.

\bibitem{HJ}
Pavel Hrubes and Pushkar~S. Joglekar, \emph{On read-\emph{k} projections of the determinant}, Electron. Colloquium Comput. Complex. \textbf{{TR24-125}} (2024).

\bibitem{karp1975computational}
Richard~M. Karp, \emph{On the computational complexity of combinatorial problems}, Networks \textbf{5} (1975), no.~1, 45--68.

\bibitem{LR}
Joseph~M. Landsberg and Nicolas Ressayre, \emph{Permanent v. determinant: An exponential lower bound assuming symmetry}, Proceedings of the 2016 {ACM} Conference on Innovations in Theoretical Computer Science, Cambridge, MA, USA, January 14-16, 2016 (Madhu Sudan, ed.), {ACM}, 2016, pp.~29--35.

\bibitem{mahajan1999determinant}
Meena Mahajan and V~Vinay, \emph{Determinant: Old algorithms, new insights}, SIAM journal on Discrete Mathematics \textbf{12} (1999), no.~4, 474--490.

\bibitem{malod2008characterizing}
Guillaume Malod and Natacha Portier, \emph{Characterizing {V}aliant's algebraic complexity classes}, Journal of complexity \textbf{24} (2008), no.~1, 16--38.

\bibitem{mignon2004quadratic}
Thierry Mignon and Nicolas Ressayre, \emph{A quadratic bound for the determinant and permanent problem}, International Mathematics Research Notices (2004), no.~79, 4241--4253.

\bibitem{mulmuley2011p}
Ketan~D Mulmuley, \emph{On {P} vs. {NP} and geometric complexity theory}, Journal of the ACM \textbf{58} (2011), no.~2, 1--26.

\bibitem{shpilka2010arithmetic}
Amir Shpilka and Amir Yehudayoff, \emph{Arithmetic circuits: A survey of recent results and open questions}, Foundations and Trends in Theoretical Computer Science \textbf{5} (2010), no.~3--4, 207--388.

\bibitem{toda1991pp}
Seinosuke Toda, \emph{{PP} is as hard as the polynomial-time hierarchy}, SIAM Journal on Computing \textbf{20} (1991), no.~5, 865--877.

\bibitem{valiant1979completeness}
Leslie~G. Valiant, \emph{Completeness classes in algebra}, STOC, 1979, pp.~249--261.

\bibitem{valiant1979complexity}
\bysame, \emph{The complexity of computing the permanent}, Theoretical computer science \textbf{8} (1979), no.~2, 189--201.

\bibitem{gathen}
Joachim von~zur Gathen, \emph{Permanent and determinant}, 27th Annual Symposium on Foundations of Computer Science, Toronto, Canada, 27-29 October 1986, {IEEE} Computer Society, 1986, pp.~398--401.

\end{thebibliography}

\end{document}